\documentclass[aps,twocolumn,pra,showpacs,superscriptaddress,floatfix]{revtex4-2}
\usepackage{amssymb,amsmath,amsfonts}
\usepackage{epsfig,graphicx}
\usepackage{hyperref}
\usepackage{amsthm}
\hypersetup{colorlinks,citecolor=NavyBlue,filecolor=black,linkcolor=NavyBlue,urlcolor=NavyBlue}

\newcommand\ha{\hat{a}}
\newcommand\hadag{\hat{a}^{\dagger}}
\newcommand\hadaga{\hat{a}^{\dagger} \hat{a}}

\newcommand\ket[1]{\left|\textstyle{#1}\right\rangle}
\newcommand\bra[1]{\left\langle\textstyle{#1}\right|}

 \newtheorem{theorem}{Theorem}

\usepackage{color}
\definecolor{blueRP}{rgb}{0.0, 0.58, 0.71}
\definecolor{blueCh}{rgb}{0.2, 0.4, 0.7}

\usepackage[dvipsnames]{xcolor}

\begin{document}
\title{Comment on ``Properties and dynamics of generalized squeezed states''}

\author{Rub{\'e}n Gordillo-Hachuel}
\affiliation{Department of Physics, Universidad Carlos III de Madrid, Avda. de la Universidad 30, Legan\'es, 28911  Madrid, Spain}
\affiliation{Department of Electronic Technology, Universidad Carlos III de Madrid, Avda. de la Universidad 30, Legan\'es, 28911  Madrid, Spain}
\author{Ricardo Puebla}
\affiliation{Department of Physics, Universidad Carlos III de Madrid, Avda. de la Universidad 30, Legan\'es, 28911  Madrid, Spain}

\begin{abstract}
    \noindent
    A recent article [\href{https://iopscience.iop.org/article/10.1088/1367-2630/add7fc}{S. Ashhab and M. Ayyash, \textit{New J. Phys. \textbf{27}, 054104} (2025)}] has reported unexpected oscillatory dynamics in generalized squeezed states of order higher than two as their squeezing parameter increases. This behaviour, observed through numerical simulations using truncated bosonic annihilation and creation operators, appeared in several properties of these states, including their average photon number. The authors argued that these oscillations reflect a genuine physical effect. Here, however, we demonstrate that the observed oscillatory behaviour is a consequence of numerical artefacts. A numerical analysis reveals that the oscillations are highly sensitive to the truncation of the Fock basis, indicating a lack of convergence. This is further supported by a theoretical analysis of the Taylor series of the average photon number, suggesting that these generalized squeezed states contain infinite energy after a finite value of the squeezing parameter. Finally, we provide an analytical proof that the average photon number of any generalized squeezed state is a non-decreasing function, thereby ruling out the possibility of intrinsic oscillatory dynamics. We hope these results help clarify the origin of the reported oscillations and highlight the special care required when dealing with high-order squeezing states.
\end{abstract}

\maketitle
\section{Introduction}

Recently, S. Ashhab and M. Ayyash~\cite{Ashhab_2025} have  argued that generalized squeezed states of order higher than two feature oscillatory dynamics, as for example in their average photon number, as a function of generalized higher-order squeezing parameter. This oscillatory behaviour contrasts with the well-known coherent and squeezed states, which exhibit an increasing photon number for increasing first-order and second-order generalized squeezing parameter, respectively.  In Ref.~\cite{Ashhab_2025}, it is asserted that such oscillatory behaviour is physical and does not arise from numerical artefacts. This conclusion is supported by means of numerical simulations with truncated bosonic operators including $N$ Fock states, showing how the frequency of the oscillations remains stable despite an increasing amplitude for growing $N$. 

In this comment we demonstrate that the oscillations are, in fact, a consequence of numerical artefacts, and do not correspond to a genuine physical effect of generalized squeezed states. In Sec.~\ref{s:num}, we begin with a numerical analysis of the average photon number for tri- and quadri-squeezed states (generalized squeezed states of order three and four, respectively), revealing a dramatic change in the oscillatory behaviour when computed in truncated spaces with similar number of Fock basis. This effect indicates lack of numerical convergence beyond some squeezing parameter, which is in agreement with the previously reported divergences of such generalized squeezed states~\cite{Braunstein1987,Hillery1990}. In Sec.~\ref{s:theo}, we investigate the convergence of the Taylor series of the average photon number for $n=3$ and $4$. This analysis  suggests that the reported oscillations take place when such tri- and quadri-squeezed states are ill defined, i.e. for values for which these states contain infinite energy despite being normalizable~\cite{Hillery1990}. We then provide a proof showing that the average photon number of any generalized squeezed state is a non-decreasing function of its high-order squeezing parameter within its radius of convergence. These results show that the reported oscillatory behaviour is a consequence of an artificially truncated Fock basis when the states are divergent, and thus the oscillations cannot be attributed to a physical effect of generalized squeezed states.

\section{Numerical Analysis}\label{s:num}

    The well-known displacement and squeezing operators can be generalized to high-order photon squeezing \cite{Fisher1984,Ashhab_2025}, corresponding to $n$-photon down conversion, 
    \begin{align}
        \hat{U}_n(r) = \exp{\{r \ha^{\dagger, n} - r^* \ha^n\}}.
    \end{align}
    \noindent
    Here, $\ha$ and $\hadag$ are the photon annihilation and creation operators, respectively, so that $[\ha,\hadag]=1$, and $r$ corresponds to the high-order squeezing parameter. Although in general $r\in\mathbb{C}$, for the analysis of average photon number $\langle \hadaga \rangle$, one can restrict without loss of generality to non-negative real values, $r\geq 0$. It is easy to verify that taking $n = 1$, we recover the displacement operator expression, while considering $n = 2$, we obtain a general expression of the two-photon squeezing operator. Our numerical analysis will focus on cases $n > 2$, specifically for $n = 3$ and $n = 4$, as also considered in Ref.~\cite{Ashhab_2025}.

    We start by defining the nth-order squeezed state as, 
    \begin{align}
        \ket{r_n} \equiv \hat{U}_n(r) \ket{0} = e^{r (\ha^{\dagger, n} -  \ha^n)} \ket{0},
    \end{align}
    \noindent
    which has an average photon number, 
    \begin{align}
        \langle \hadaga \rangle_n \equiv \bra{r_n} \hadaga \ket{r_n}.
    \end{align}
    \noindent
    For the cases $n = 3$ and $n = 4$, we refer to the states $\ket{r_3}$ and $\ket{r_4}$ as tri-squeezed and quadri-squeezed states, respectively. Now, we numerically calculate the average photon number of these states, as a function of the squeezing parameter $r$, for different truncation $N$ of the Fock basis. The results for both cases are shown in Figs.~\ref{fig:n_vs_r}(a) $(n=3)$ and \ref{fig:n_vs_r}(b) $(n=4)$, and are consistent with those presented in Ref.~\cite{Ashhab_2025}. For sufficiently small values of $r$, the average photon number for both cases  does not change with the truncation. As $r$ grows, however, the curves for similar truncations match albeit their slope increases with $N$  without appearing to converge to any finite value. This in agreement with~\cite{Hillery1990} where it was shown that after a finite value of $r$, the average photon number diverges. For even higher values of the squeezing parameter, the average photon number begins to oscillate, showing significantly different curves for closely spaced values of $N$. These distinct shapes depend on the parity of $N'=\lfloor (N-1)/n\rfloor$ with $\lfloor x\rfloor=\max \{n\in\mathbb{N}:n\leq x\}$,  revealing  how the enlargement of the Hilbert space affects the system when the  newly added Fock state can be populated, i.e. when $(N-1)/n\in\mathbb{N}$. These behaviours, contrary to what is stated in Ref.~\cite{Ashhab_2025}, strongly suggests that the results in these regions are numerical artefacts lacking physical significance. This includes the average photon number $\langle \hadaga \rangle_n$, as well as the Husimi and Wigner functions. In the next section we will test this hypothesis by performing an analytical analysis of these functions, estimating the radius of convergence and proving that for the physical regime, i.e. for $r$ such that $\langle \hadaga\rangle_n<\infty$, the average photon number is a non-decreasing function of $r$. 

    \begin{figure}[t]
        \centering
        \includegraphics[width=\linewidth]{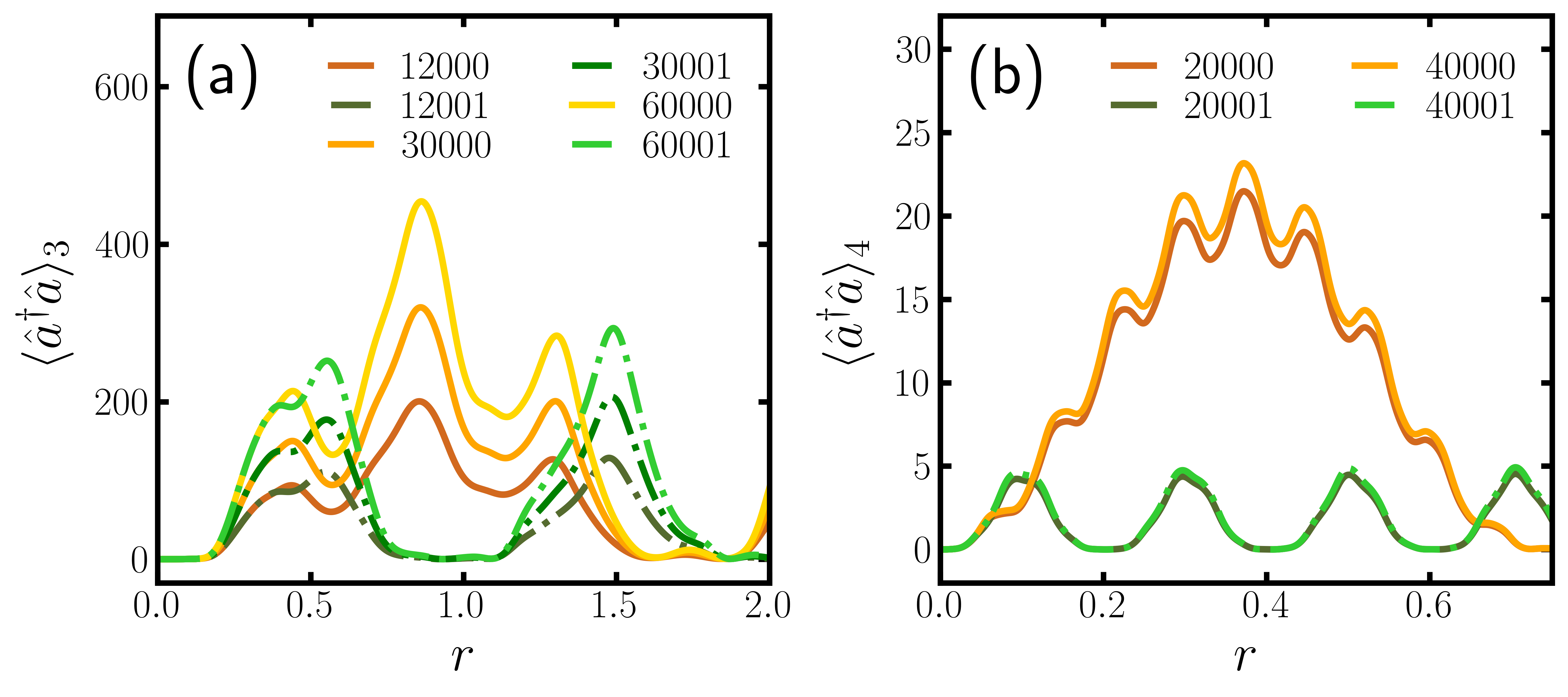}
        \caption{Behaviour of the average photon number $\langle \hadaga \rangle_n$ as a function of the squeezing parameter $r$. Panels (a) and (b) corresponds to the tri-squeezed, $\ket{r_3} = \hat{U}_3(r) \ket{0}$, and  quadri-squeezed, $\ket{r_4} = \hat{U}_4(r) \ket{0}$, respectively. Different curves have been obtained under different truncation $N$ of the Fock basis, as indicated in the labels. Note the significant change of $\langle \hadaga\rangle_n$ for similar $N$ that suggests lack of convergence.}
        \label{fig:n_vs_r}
    \end{figure}

\section{Convergence and monotonic analysis}\label{s:theo}
    
    Let us now show that the oscillations reported in Ref.~\cite{Ashhab_2025}, and reproduced here in Fig. \ref{fig:n_vs_r}, are numerical artefacts as a consequence of the divergent nature of these states for large $r$ values. We  first provide a theoretical analysis of the divergence of the average photon number for tri- and quadri-squeezed states, and then show that for a generic multisqueezed state, $\langle \hadaga\rangle_n$ displays a global minimum at $r=0$ without local maxima as a function of $r$, thereby ruling out the physical significance of the oscillations. 

    \begin{figure}
        \centering
        \includegraphics[width=\linewidth]{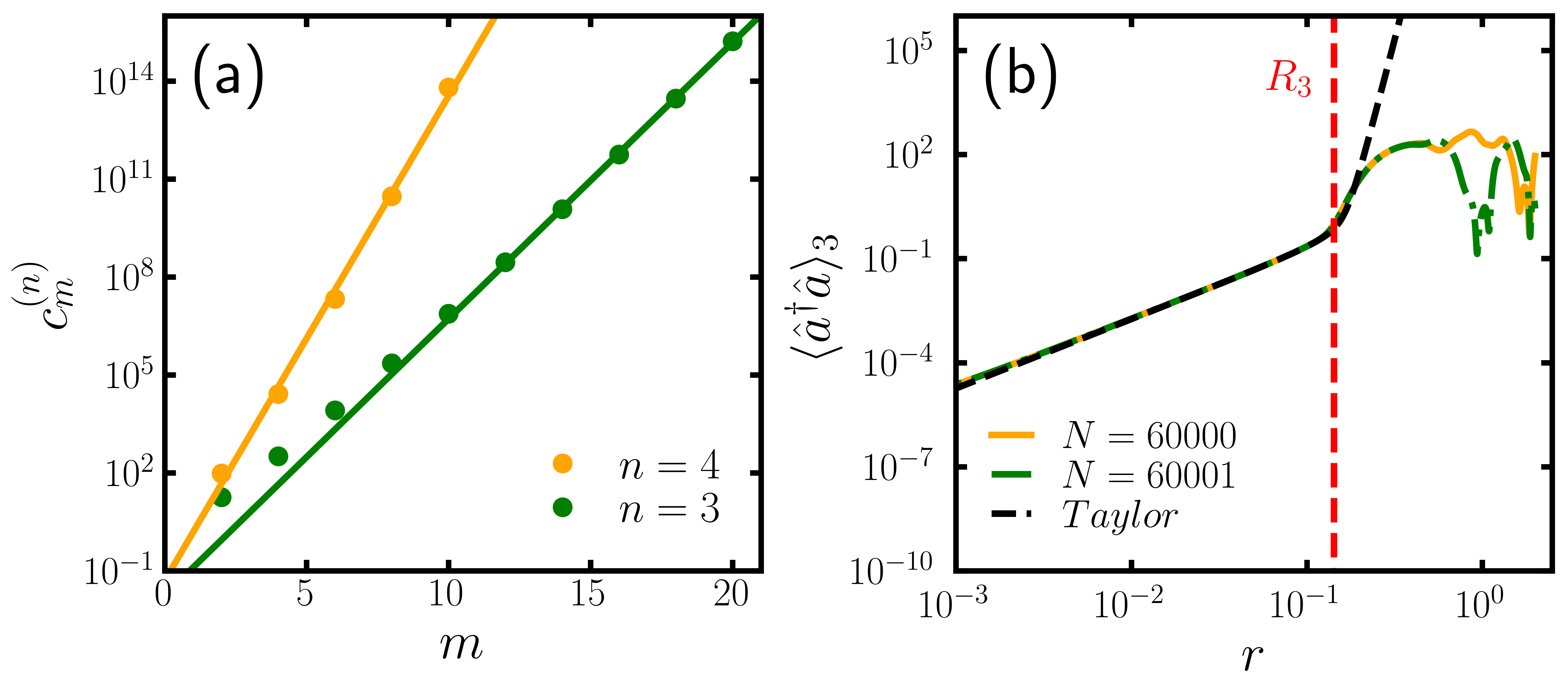}
        \caption{(a) First $M=20$ and $M=10$ coefficients $c_m^{(3)}$ (orange dots) and $c_m^{(4)}$ (green dots) for $n = 3$ and  $n = 4$, respectively, as a function of $m$. Note that, by symmetry, the odd coefficients are zero $c_{2m+1}^{(n)}=0$. An exponential fit for the last five points for each case have been plotted as straight lines. (b) Average photon number for the tri-squeezed state as a function of the squeezing parameter for different approximations. Orange and green lines correspond to a numerical computation using $N=60000$ and $N=60001$ Fock basis, respectively. The dashed black line represents the Taylor series considering the first $M=20$ terms. The dashed red line denotes the estimated radius of convergence $R_3$ for the Taylor series.}
        \label{fig:coefficients}
    \end{figure}
    
    The average photon number of a multisqueezed state can be written as
    \begin{align}
        \langle \hadaga\rangle_n&=\bra{0}\hat{U}_n^\dagger(r)\hadaga \hat{U}_n(r)\ket{0}\\\nonumber &=\bra{0}\left[\sum_{m=0}^\infty \frac{r^m}{m!}[(\hat{a}^{\dagger,n}-\hat{a}^n),\hadaga]_m \right]\ket{0},
    \end{align}
    where $[\hat{A},\hat{B}]_{m+1}=[A,[A,B]_{m}]$ is the nested commutator of order $m$ and $[A,B]_{m=0}=B$. This allows us to write the average photon number as infinite series,
    \begin{align}\label{eq:adaga_m}
    \langle \hadaga\rangle_n=\sum_{m=0}^\infty c_m^{(n)} r^m,
    \end{align}
    where the coefficients $c_m^{(n)}$ result from $c_m^{(n)}=\frac{1}{m!}\bra{0}[(\hat{a}^{\dagger,n}-\hat{a}^n),\hadaga]_m\ket{0}$. For $n \geq 3$, the higher-than-quadratic boson operators prevents a closed expression for Eq.~\eqref{eq:adaga_m}. Yet, one can compute the coefficients $c_m^{(n)}$ up to some order $M$.  In Fig.~\ref{fig:coefficients}(a), we present the first $M=20$ and $M=10$ coefficients $c_m^{(3)}$ and $c_m^{(4)}$ for $n = 3$ and  $n = 4$, respectively, as a function of $m$. Note that, by symmetry, the odd coefficients are zero $c_{2m+1}^{(n)}=0$. The coefficients have been computed using the \textsc{Mathematica}-based software DiracQ~\cite{DiracQ}. As it can observed in Fig.~\ref{fig:coefficients}(a), the coefficients grow exponentially with $m$, i.e. $c_m^{(n)}\propto e^{\alpha_n m}$. A fit, represented as straight lines in the graph, allows us to obtain the coefficients $\alpha_n$, resulting in $\alpha_3=1.95(2)$ and $\alpha_4=3.4(1)$. Therefore, we can estimate the radius of convergence of the average photon number as $R^{-1}_n=\lim_{m\rightarrow \infty} |c_m^{(n)}|^{1/m}=e^{\alpha_n}$, which results in $R_3\approx 0.14$ and $R_4\approx 0.03$. These approximated $R_{3,4}$ also apply to the results shown in  Ref.~\cite{Ashhab_2025}, as we used the same definition for the generalized squeezing operator $\hat{U}_n(r)$. In Fig.~\ref{fig:coefficients}(b) we have represented the average photon number for the tri-squeezed state, using the numerical approximation employing a truncated space, and the Taylor expansion keeping the first $M=20$ terms (cf. Eq.~\eqref{eq:adaga_m}). The radius of convergence $R_3$ is also shown as a vertical line. For values $r>R_3$, the state contains an infinite number of photons, and therefore it becomes non-physical~\cite{Hillery1990}. Due to the divergence for $r>R_n$, results out of the convergence radius using artificially truncated Fock basis or series are misleading numerical artefacts that lack physical meaning, including the Husimi and Wigner functions shown in \cite{Ashhab_2025}. This unphysical regime indicates that the simple interaction assumed to generate the gate $\hat{U}_{3}(r)$ breaks down, as for example requiring to explicitly include the quantized pump field in a $3$-photon downconversion as discussed in Ref.~\cite{Hillery1990}. We note that the oscillations take place in this regime. For $r<R_3$, however, these different approximations to the average photon number match, as it can be seen in Fig.~\ref{fig:coefficients}(b), showing a growing number of photons for increasing $r$. This corresponds to the physical region. Although not explicitly shown, the situation for $n=4$ is equivalent.
    
    We can complement this observation by showing that the average photon number for any multisqueezed state is a non-decreasing function with $r$.  For that, one can compute the second derivative of $\langle \hadaga\rangle_n$ with respect to $r$, which can be written as
    \begin{align}
        \frac{d^2}{dr^2} \langle \hadaga \rangle_n &= 2n \bra{r_n} [\ha^n, \ha^{\dagger, n}] \ket{r_n}.
    \end{align}
    Let us define the commutator in the previous expression as $\hat{A}_n\equiv [\hat{a}^n,\hat{a}^{\dagger, n}]$, that  can be expressed in a compact form as~\cite{Cahill1969}
    \begin{align}
        \hat{A}_n = \sum_{k=1}^{n} k! C_{n,k}^2 (\hat{a}^{\dagger, n-k} \hat{a}^{n-k}), \ {\rm for}\ n\geq 1,
    \end{align}
    \noindent
    where $C_{n,k}$ represents the binomial coefficient, $C_{n,k} = \frac{n!}{k!(n-k)!}$. Since $(\hat{a}^{\dagger, n-k} \hat{a}^{n-k})$ with $n-k \in \mathbb{N}$ is non-negative, the operator $\hat{A}_n$ is strictly positive, $\hat{A}_n>0,\ \forall n\in\mathbb{N}^+$. Note that one can also see that $\hat{A}_n>0$ by expressing the commutator $\hat{A}_n$ in the Fock basis. In this manner, we arrive to
    \begin{align}
        \frac{d^2}{dr^2} \langle \hadaga \rangle_n > 0, \quad \forall r\geq 0, \ \forall n\in \mathbb{N}^+.
    \end{align}
    As the second derivative is a strictly positive quantity, one can ensure that $\langle \hadaga \rangle_n$ does not display local maxima. As commented above, given that $\langle \hadaga\rangle$ is invariant under any boson rotation of the form $\hat{a}\rightarrow e^{i\varphi}\hat{a}\  {\rm for} \ \varphi\in[0,2\pi)$, we can extend this to $|r|$. Therefore, given that $\langle \hadaga\rangle_n$ 
    has a minimum at $|r| = 0$, i.e. $\bra{r_n=0}\hadaga\ket{r_n=0}=0$, it follows that $\langle \hadaga \rangle_n$ displays an absolute minimum at $|r| = 0$, and has no local maxima. This can be formulated as follows.
    
    \begin{theorem} Generalized multisqueezed states $\ket{r_n} = \hat{U}_n(r)\ket{0}$, parametrised by $r\in\mathbb{C}$ with $\hat{U}_n(r)=e^{(r\ha^{\dagger,n}-r^*\ha^n)}$ and $n\in\mathbb{N}^+$, contain an average photon number $\langle \hadaga \rangle_n \equiv \bra{r_n} \hadag \ha \ket{r_n} = f_n(|r|)$, being $f_n(|r|)$  a non-decreasing function in its domain $0\leq |r|<R_n$. 
    \end{theorem}
    \begin{proof}
    Assume  $\langle \hadaga \rangle_n$ to be a non-monotonic function with $|r|$. This means that $\langle \hadaga \rangle_n$ decreases after some value $|\tilde{r}|>0$, reaching a local maximum at $|\tilde{r}|$. However, $\frac{d^2}{d|r|^2}\langle \hadaga\rangle_n >0 \ \forall |r|$, and so also at $|\tilde{r}|$. This rules out the existence of local maxima, therefore contradicts the first assumption, i.e. $\langle \hadaga \rangle_n$ is a monotonic function. Moreover, since  $\langle \hadaga \rangle_n\geq 0, \forall r\in\mathbb{C},\forall n\in\mathbb{N}^+$, being zero for $|r|=0$, implies that $\langle \hadaga\rangle_n$ is a non-decreasing function with $|r|$. 
    \end{proof}

\section{Conclusion}
In a recent article~\cite{Ashhab_2025}, it has been argued that generalized squeezed states of order $n>2$ exhibit intrinsic oscillations, as for example on their average photon number as a function of the squeezing parameter. Such unexpected oscillatory behaviour is absent in the standard coherent and $n=2$ squeezed states. In Ref.~\cite{Ashhab_2025}, these oscillations have been attributed to a genuine physical effect of such high-order generalized squeezed states. Here we show that such oscillatory behaviour is a numerical artefact that arises due to the divergence of photon occupation for a large value of the squeezing parameter. We start with a numerical analysis computing the average photon number under different truncation $N$ of the Hilbert space. This analysis reveals that the oscillations change significantly for similar $N$, already suggesting lack of numerical convergence. A further theoretical analysis allows us to estimate the radius of convergence of the Taylor series of the average photon number  for the $n=3$ and $n=4$ squeezed states, referred to as tri- and quadri-squeezed. Our results show that the oscillations reported in Ref.~\cite{Ashhab_2025} take place in a non-physical region, i.e. where the states contain an infinite number of photons. The oscillations are therefore a consequence of an artificially truncated Hilbert space. Furthermore, we demonstrate that in the physical region, any $n$-order generalized squeezed state  displays an average photon number that is a non-decreasing function with the squeezing parameter, thereby ruling out any oscillatory behaviour. We hope our results help clarify the origin of the reported oscillations and highlight the special care required when dealing with high-order squeezing states.

\acknowledgments
We acknowledge financial support form the Spanish Government via the project TSI-069100-2023-8 (Perte Chip-NextGenerationEU). RP acknowledges the Ram{\'o}n y Cajal (RYC2023-044095-I) research fellowship.


%

\end{document}